\documentclass[12pt]{amsart}

\usepackage{latexsym}
\usepackage{amsmath}
\usepackage{epsfig}
\usepackage{epic}
\usepackage{amssymb}

  \usepackage[text={6.6in,9in},centering]{geometry}

\newtheorem{theorem}{Theorem}[section]
\newtheorem{lemma}[theorem]{Lemma}
\newtheorem{Claim}[theorem]{Claim}

\newtheorem{proposition}[theorem]{Proposition}

\newcommand{\cQ}{{\mathcal Q}}
\newcommand{\cF}{{\mathcal F}}

\newcommand{\TA}[4]{{#1}\ominus {#2}|{#3} \ominus {#4}}
\newcommand{\TB}[4]{{#1}\oplus{#2}|{#3}\ominus {#4}}
\newcommand{\TBB}[4]{{#1}\ominus{#2}|{#3}\oplus{#4}}
\newcommand{\TC}[4]{{#1}\oplus{#2}|{#3}\oplus{#4}}
\newcommand{\TD}[4]{{#1}\oplus{#2}\oplus{#3}\oplus{#4}}

\newcommand{\oalpha}{\ast}
\newcommand{\obeta}{\diamond}
\newcommand{\ogamma}{\circ}

\newcommand{\closure}{{\rm cl}_2}
\newcommand{\med}{\mathop{med}}

\begin{document}

\title{Quarnet inference rules for level-1 networks}


\author{Katharina T.\ Huber}
\address{School of Computing Sciences, University of East Anglia, Norwich, U.K.}
\email{k.huber@uea.ac.uk}

\author{Vincent Moulton}
\address{School of Computing Sciences, University of East Anglia, Norwich, U.K.}
\email{v.moulton@uea.ac.uk}

\author{Charles Semple}
\address{School of Mathematics and Statistics, University of Canterbury, Christchurch, New Zealand}
\email{charles.semple@canterbury.ac.nz}

\author{Taoyang Wu}
\address{School of Computing Sciences, University of East Anglia, Norwich, U.K.}
\email{taoyang.wu@uea.ac.uk}



\date{\today}

\begin{abstract}
An important problem in phylogenetics is the construction of phylogenetic trees. 
One way to approach this problem, known as the supertree method, 
involves inferring a phylogenetic tree with leaves consisting of a set $X$ of species from  
a collection of trees, each having leaf-set some subset of $X$. In 
the 1980's characterizations, 
certain inference rules were given for when a collection 
of 4-leaved trees, one for each 4-element subset of $X$, can 
all be simultaneously displayed by a single supertree with leaf-set $X$.  Recently, 
it has become of interest to extend such results to phylogenetic networks. These
are a generalization of phylogenetic trees which can be used to represent 
reticulate evolution (where species can come together to form a new species).
It has been shown that a certain type of phylogenetic network, called a
level-1 network, can essentially be constructed from 4-leaved trees. However, the problem of
providing appropriate inference rules for such networks remains unresolved. Here we show that 
by considering 4-leaved networks, called quarnets, as opposed to 4-leaved
trees, it is possible to provide such rules. In particular, we show that  
these rules can be used to characterize when a collection 
of quarnets, one for each 4-element subset of $X$, can 
all be simultaneously displayed by a level-1 network with leaf-set $X$. 
The rules are an intriguing mixture of tree inference rules, and
an inference rule for building up a cyclic ordering of $X$ from orderings on subsets of $X$ of 
size 4. This opens up several new directions of research for inferring phylogenetic networks from 
smaller ones, which could yield new algorithms for solving the supernetwork problem in phylogenetics.   
\end{abstract}

\maketitle


\section{Introduction}

One of the main goals in phylogenetics is to develop methods for 
constructing evolutionary trees, the tree-of-life being a prime example of
such a tree \cite{Letunic}.
Mathematically speaking, for a set $X$ of species, a phylogenetic $X$-tree
is a (graph theoretical) tree 
with leaf set $X$ and no degree-$2$ vertices; it is {\em binary} if every
internal vertex has degree three.
A popular approach to constructing such trees, called the  {\em supertree method},  
is to build them up from smaller trees \cite{emonds}. The smallest possible trees that 
can be used in this approach are {\em quartet trees}, that is, binary
phylogenetic trees having 4 leaves
(see e.g. Figure~\ref{illustrate} for the quartet tree $ab|cd$ with
leaf-set $\{a,b,c,d\} \subseteq X$).
Thus it is natural to ask the following question: How should we decide whether or not 
it possible to simultaneously  display all of the quartet trees in a given collection 
$\mathcal Q$ of quartet trees by some phylogenetic tree?

\begin{figure}[ht]
	\centering
	\includegraphics[scale=0.7]{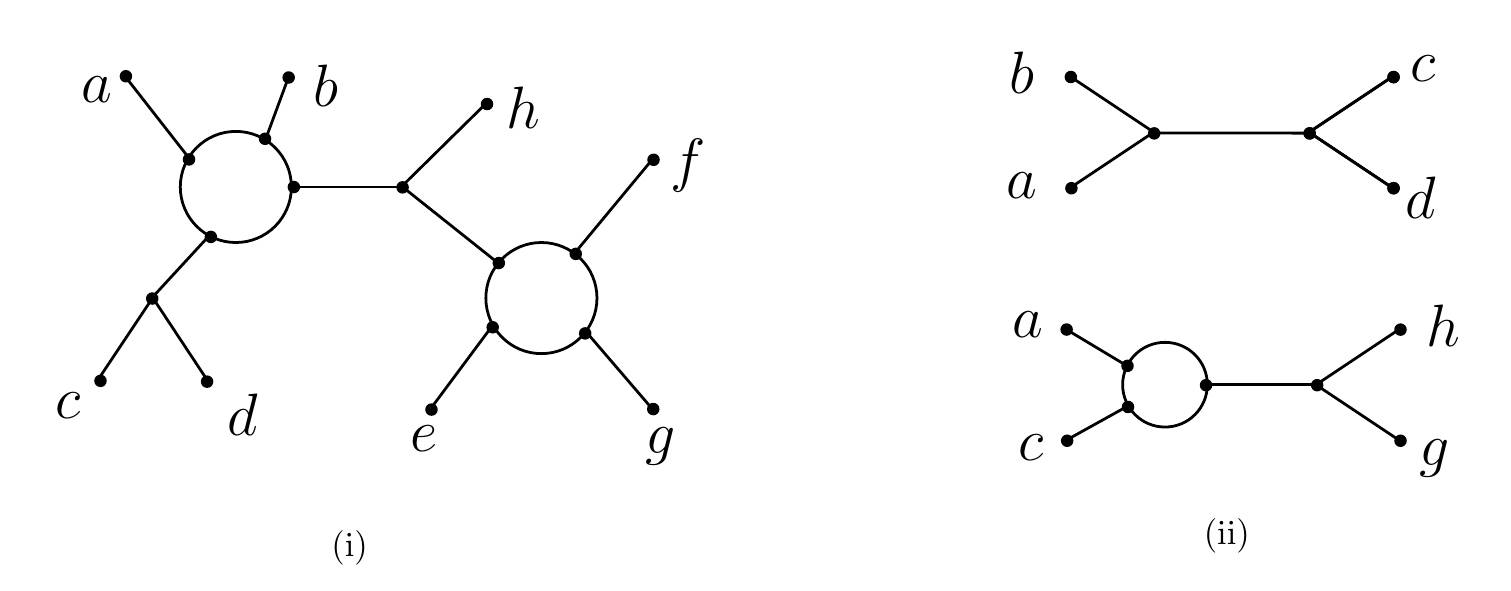}
	\caption{(i) A level-1 phylogenetic network with leaf-set $X=\{a,b,\dots,h\}$. (ii) Top: a quartet
	tree with leaf-set $\{a,b,c,d\}$, also denoted by $ab|cd$. 
	Bottom: a quarnet with leaf-set $\{a,c,h,g\}$. Both the quartet tree and quarnet 
	are displayed by the level-1 network in (i).}
	\label{illustrate}
\end{figure}
	
In case the collection $\mathcal Q$ consists of a quartet tree for every 
possible subset of $X$ of size 4 (which we denote by $X \choose 4$), 
this problem has an elegant solution that was originally 
presented by Colonius and Schulze in 1981 \cite{C81}
(see also \cite{bandelt1986reconstructing} for related results).  
We present full-details in Theorem~\ref{thm:quartet:tree} below, but essentially their result states that, 
given a collection of quartet trees $\mathcal Q$, one for each element in $X \choose 4$, there 
exists (a necessarily unique) binary phylogenetic $X$-tree displaying every quartet tree in the collection
if and only if when the quartet trees $ab|cx$ and $ab|xd$ are contained in $\mathcal Q$ 
then so is the quartet tree $ab|cd$. Rules such as $ab|cx$ plus $ab|xd$ implies 
$ab|cd$ are known as {\em inference rules}, and
they have been extensively studied in the phylogenetics 
literature (see e.g. \cite[Chapter 6.7]{SS03}).
	
Although phylogenetic trees are extremely useful for representing 
evolutionary histories, in certain circumstances they can be inadequate.
For example, when two viruses recombine to form a new virus (e.g. swine flu), 
this is not best represented by a tree as it involves species combining together
to form a new one rather than splitting apart. In such cases, phylogenetic networks 
provide a more accurate alternative to trees and there has been much recent work on such 
structures (see e.g. \cite[Chapter 10]{S16} for a recent review). 

In this paper, we will consider properties of a particular type of phylogenetic 
network called a {\em level-1 network} \cite{G12}. For a set $X$ of species, this is
a connected graph with leaf-set $X$ and such that 
every maximal subgraph with no cut-edge is either a vertex or a cycle
(see Section~\ref{sec:prelim} for more details). Our main results will apply to 
binary level-1 networks, where we also assume
that every vertex has degree 1 or 3. We present an example 
of such a network in Figure~\ref{illustrate}.
Note that a phylogenetic $X$-tree is a special example of a level-1 network 
with leaf-set $X$.  As with phylogenetic $X$-trees it is possible to 
construct level-1 networks from quartets \cite{G12}. However, it has been
pointed out that there are problems with understanding
such networks in terms of inference rules (see e.g. \cite[p.2540]{kp14}).

Here, we circumvent these problems by considering a certain 
type of subnetwork of level-1 network called a {\em quarnet} instead
of using quartet trees. 
A quarnet is a 4-leaved, binary, level-1 network (see e.g. Figure~\ref{illustrate}); they are 
displayed by binary level-1 networks in a similar way to quartets (see Section~\ref{sec:qnet} for details). 
As we shall see, quarnets naturally lead to inference rules for level-1 networks which can 
be thought of as a combination of quartet inference and inference rules for
building circular orderings of a set. Moreover, in our main result we show that, just as 
with phylogenetic trees, the quarnet inference rules that we introduce
can be used to characterize when a collection of quarnets, one for
each element in ${X \choose 4}$, can be displayed by a binary level-1 network
with leaf-set $X$. 

We now summarize the contents of the rest of the paper.
In the next section we present some preliminaries
concerning phylogenetic trees and level-1 networks, as well
as their relationship with quartets. Then, in Section~\ref{sec:qnet}, 
we prove an analogous theorem to the quartet results of
Colonius and Schulze for level-1 networks (Theorem~\ref{thm:quarnet}).
In Section~\ref{sec:characterize}, we use Theorem~\ref{thm:quarnet} to provide 
a characterization for when a set of quartets, one for each element 
of $X \choose 4$, can be displayed by a binary level-1 network (Theorem~\ref{thm:quartet}).
In Section~\ref{sec:closure}, we then define the closure
of a set of quarnets. This can be thought of as the collection of 
quarnets that is obtained by applying inference rules to a given collection of quarnets 
until no further quarnets are generated.
We show that this has similar properties
to the so-called {\em semi-dyadic closure} of a set of quartets (see Theorem~\ref{thm:closure}). 
We conclude with a brief discussion of some possible further directions.

\section{Preliminaries}\label{sec:prelim}

In this section, we review some definitions as well as results concerning the connection between phylogenetic trees and quartets. From 
now on, we assume that $X$ is a finite set with $|X|\ge 2$.

\subsection{Definitions}

An {\em unrooted phylogenetic network $N$ (on $X$)} 
(or {\em network $N$ (on $X$)} for short) is a connected graph $(V,E)$ with 
$X \subseteq V$, every vertex has either degree 1 or degree at least 3, 
and the set of degree-$1$ vertices is $X$. 
The elements in $X$ are the {\em leaves} of $N$. We also
denote the leaf-set of $N$ by $L(N)$. The network is called {\em binary} if every vertex 
in $N$ has degree 1 or 3. An {\em interior vertex} of $N$ is a vertex that is not a leaf. 
A {\em cherry} in $N$ is a pair of leaves that are adjacent with 
the same vertex. Two phylogenetic networks $N$ and $N'$ on $X$ are {\em isomorphic}
if there exists a graph-theoretical isomorphism between $N$ and $N'$ whose
restriction to $X$ is the identity map.

Note that a {\em phylogenetic ($X$-) tree} is a network which is also a tree.
For any three vertices $u_1,u_2,u_3$ in such a tree $T$, 
their {\em median}, denoted by $\med(u_1,u_2,u_3)=\med_T(u_1,u_2,u_3)$, 
is the unique vertex in $T$ that is contained in every path between $u_1$, $u_2$ and $u_3$. 

A \emph{cut-vertex} of a network is an vertex
whose removal disconnects the network, and a \emph{cut-edge} 
of a network is an edge whose removal disconnects the network.  
A cut-edge is \emph{trivial} if one of the connected components 
induced by removing the cut-edge is a vertex (which must necessarily be a leaf).
A network is  \emph{simple} if all of the cut-edges are trivial 
(so for instance, note that phylogenetic trees with 
more than three leaves are \emph{not} simple networks).
A network $N$ is {\em level-1} if every maximal subgraph in $N$ that has no cut-edge is either a vertex or a cycle.
Note that we shall say that a network $N$ on $X$, where $|X|\ge 3$, 
is of {\em cycle-type} if it contains a unique cycle of length $|X|$, and the 
number of vertices in $N$ is $2|X|$ (so in particular, a network
is of cycle-type if it is simple, binary, level-1 and is not a phylogenetic tree).

In what follows it will be useful to consider a certain type 
of operation on a level-1 network, which we define as follows.
For a level-1 network $N$ on $X$, let $u$ be an interior vertex of $N$ that 
is not contained in any cycle in $N$. 
Furthermore, let $(v_1, v_2, \cdots,v_k)$, where $k\ge 3$, be a
circular ordering 
of the set of vertices in $N$ that are adjacent to $u$.
Then we obtain a new network $N'$ on $X$ from $N$ 
by removing vertex $u$ and all edges incident with it
and inserting new vertices $u_i$ and
new edges $\{u_i,v_i\}$ and $\{u_i, u_{i+1}\}$ for all $1\leq i \leq k$ (see Fig.~\ref{fig:blowup}). 
Here we use the convention that $k+1$ is identified with 1. 
We say that $N'$ is obtained from $N$ by a {\it blow-up} 
operation on  $u$ (using the given circular ordering of its neighbours). 
Note that $N'$ is a level-$1$ network with one more cycle than  $N$. 
Note that blow-up operations on the 
same vertex but with different circular orderings of 
its neighbours may lead to non-isomorphic networks. We illustrate a
blow-up operation in Fig.~\ref{fig:blowup}.

\begin{figure}[ht]
\centering
\includegraphics[scale=0.8]{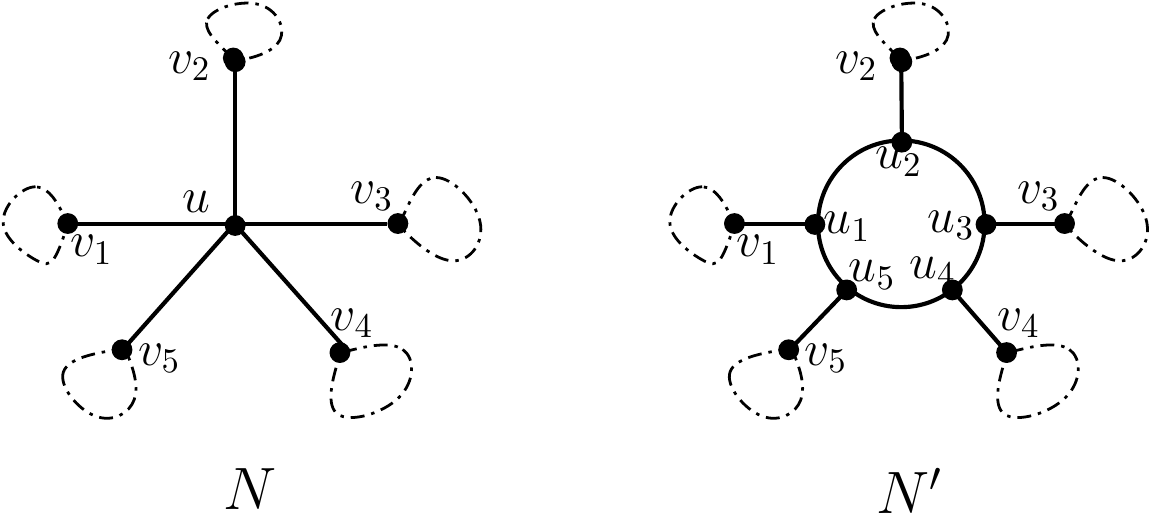}
\caption{Example of blow-up operations: $N'$ is obtained from $N$ by a blow-up operation on $u$.}
\label{fig:blowup}
\end{figure}

\subsection{Quartets, Trees and Networks}

We now briefly recall some notation and results 
concerning quartet systems (for more details see \cite[Chapter 3]{DHKMS12}). 

Although quartets are often considered as being 4-leaved trees, here it is more
convenient to consider a {\em quartet} $Q$ to be a 
partition of a subset $Y$ of $X$ of size 4 into
two subsets of size 2. The set $Y$ is called the {\em support} of $Q$.
If $Q = \{\{a,b\},\{c,d\}\}$ for $a,b,c,d \in X$ distinct, we denote $Q$ by $ab|cd$.  
The set of all quartets on $X$ is denoted by $\cQ(X)$, and 
any non-empty subset $\cQ\subseteq \cQ(X)$ is called a {\em quartet system} (on $X$).
Given a quartet system $\cQ$ on $X$ and a subset $Y\in {X \choose 4 }$, let 
$m(Y)=m_{\cQ}(Y)$ be the number of quartets in $\cQ$ whose support 
is $Y$. For simplicity, we write $m(\{a,b,c,d\})$ as $m(a,b,c,d)$. 
If $m(Y)\geq 1$ for every subset $Y\in {X \choose 4 }$, then $\cQ$ is said to be {\em dense}. 

Following the terminology in~\cite{DHKMS12}, a quartet system $\cQ$ is:
\begin{itemize} 
\item {\em thin} if no pair of quartets in $\cQ$ have the same support;
\item {\em saturated} if for all $\{a,b,c,d,x\}\in {X \choose 5}$ 
with $ab|cd \in \cQ$, the system $\cQ$ contains at least one quartet in $\{ax|cd, ab|cx\}$;
\item {\em transitive} if for all 
$\{a,b,c,d,x\}\in {X \choose 5}$, if $\{ab|cx, ab|xd\}\subseteq \cQ$ holds, then $ab|cd$ 
is also contained in $\cQ$.  
\end{itemize}

These concepts are related as follows:

\begin{lemma}
	\label{lem:transitive}
	Suppose that $\cQ$ is a quartet system on $X$. If $\cQ$ is saturated and thin, then $\cQ$ is transitive.  
\end{lemma}

\begin{proof}
	We use a similar argument to that used in~\cite[Lemma 1]{bandelt1986reconstructing}. 
	Suppose $\{a,b,c,d,x\}\in {X \choose 5}$ with $\{ab|cx, ab|xd\}\subseteq \cQ$. 
	We need to show $ab|cd\in \cQ$.
	
	Since $\cQ$ is saturated and $ab|cx$ is contained in $\cQ$, we have 
	$\{ab|cd, ad|cx\}\cap \cQ\not= \emptyset$. Using a similar argument, $ab|dx$ in $\cQ$ 
	implies that $\{ab|cd, ac|dx\}\cap \cQ\not= \emptyset$. Therefore, 
	we must have  $ab|cd\in \cQ$ as otherwise $\{ad|cx,ac|dx\}\subset \cQ$, a 
	contradiction to the assumption that $\cQ$ is thin.
\end{proof}
 
 A quartet $ab|cd$ on $X$ is {\em displayed} by a phylogenetic $X$-tree $T$  
if the path between $a$ and $b$ in $T$ is vertex disjoint from 
the path between $c$ and $d$ in $T$. The quartet system displayed by $T$ is denoted by $\cQ(T)$.
 
In view of~\cite[Theorem 3.7]{DHKMS12} and the last lemma, we 
have the following slightly stronger characterisation of quartet systems 
displayed by a phylogenetic tree, which was stated 
in~\cite[Proposition 2]{bandelt1986reconstructing} using slightly different terminology.

\begin{theorem}
\label{thm:quartet:tree}
A quartet system $\cQ\subseteq \cQ(X)$ is of the 
form $\cQ=\cQ(T)$ for a (necessarily unique) phylogenetic $X$-tree $T$ 
if and only if $\cQ$ is thin and saturated. 
\end{theorem}

We now turn our attention to the relationship between quartets and level-1 networks.

A split $A|B$ of $X$ is a bipartition of $X$ into two non-empty parts $A$ and $B$ (note that 
since $A|B$ is a bipartition, order does not matter, that is, $A|B=B|A$).
Such a split is induced by a network $N$ if there 
exists a  cut-edge in $N$ whose removal results in two connected components,
one with leaf-set $A$ and the other with leaf-set $B$.  
A quartet $ab|cd$ is {\em exhibited} by a network $N$ if 
there exists a split $A|B$ induced by $N$ such 
that $\{a,b\}\subseteq A$ and $\{c,d\}\subseteq B$.  

Note that if a quartet $ab|cd \in \cQ(X)$ is exhibited by $N$, then it is {\em displayed} 
by $N$, that is, $N$ contains two disjoint paths, one from $a$ to $b$, and 
the other from $c$ to $d$. However, the converse is not true. 
For example, quartet $ab|cd$ is displayed by 
the network in Fig.~\ref{fig:qnet}(iv), but $ab|cd$ is not exhibited by this network.
Given a network $N$, we let $\Sigma(N)$ denote the set of quartets 
exhibited by $N$, and let $\cQ(N)$ be the set of quartets displayed by $N$. 
In light of the last remark, clearly we have $\Sigma(N)\subseteq \cQ(N)$.

\section{Quarnets} \label{sec:qnet}

In this section, we shall show that an analogue of Theorem~\ref{thm:quartet:tree} holds
for quarnets and level-$1$ networks. We begin by formally defining the concept of a quarnet and how
quarnets can be obtained from level-1 networks.

Given a binary, level-1 phylogenetic network $N$ on $X$ 
and a subset $A\subseteq X$, we let $N|_A$ denote the {\em network induced on $A$ by $N$}, 
which is obtained from $N$ by deleting all edges that are not contained in some 
path between a pair of elements in $A$, removing all 
isolated vertices, and then repeatedly applying the following two
operations until neither of them is applicable (i) suppressing degree-$2$ vertices, 
and (ii) suppressing parallel edges. 
Note that $N|_A$ is a binary, level-1 phylogenetic network on $A$. 

\begin{figure}[ht]
\centering
\includegraphics[scale=1]{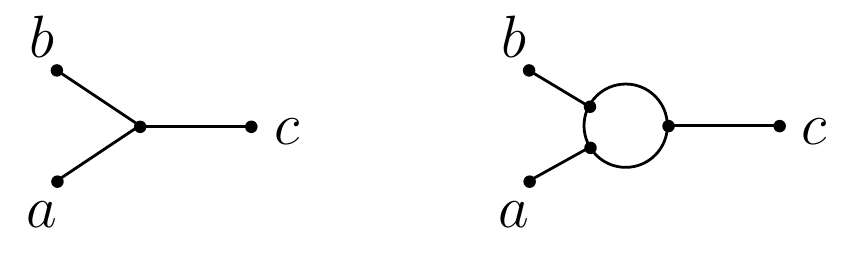}
\caption{The two types of trinets: tree type (left) and  cycle type (right).  }
\label{fig:trinet}
\end{figure}

A {\em trinet} is a binary, level-1 phylogenetic 
network on three leaves. Note that there are two types of 
trinets: one is of cycle type; the other does not contain a cycle
and is of {\em tree type} (see Fig.~\ref{fig:trinet}
for an illustration).  Similarly, a {\em quarnet or qnet} 
is a binary, level-1 phylogenetic network with four leaves.
The leaf-set $L(F)$ of a qnet $F$ is called its {\em support}.
As illustrated in Fig.~\ref{fig:qnet}, there are four types of qnets: Type I qnets contain no cycles; Type II qnets contain one cycle and one non-trival cut-edge; Type III qnets contain two cycles; and Type IV qnets contain no non-trivial cut-edge. 
A {\em qnet system} $\cF$ on $X$ is a collection of qnets all of whose 
supports are contained in  $X$.  We shall say that a 
qnet $F$ with support $A \subseteq X$ 
is {\em displayed} by a network $N$ on $X$ if $F$ is isomorphic to $N|_A$.
Moreover, 
we let $\cF(N)$ be the qnet system displayed by $N$, that is, 
$$
\cF(N)=\{N|_A\,~~\mbox{for all $A\subseteq X$ with $|A|=4$}\}.
$$ 

\begin{figure}[ht]
\centering
\includegraphics[scale=0.6]{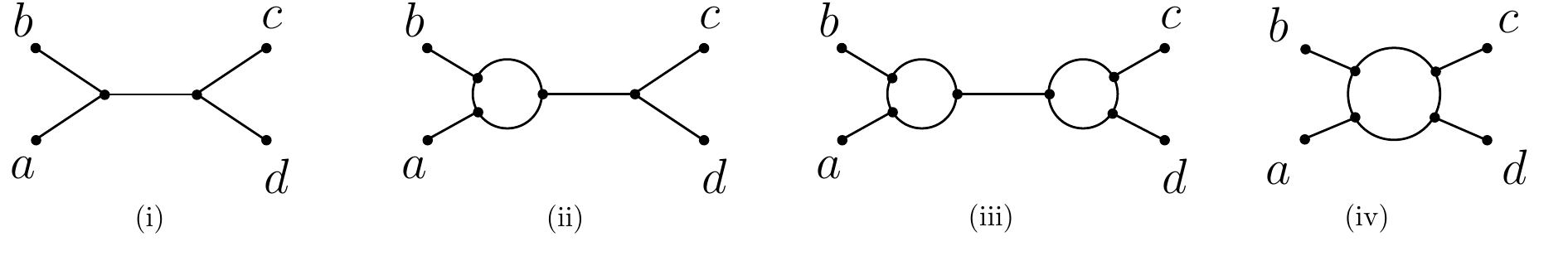}
\caption{Four types of qnets on $X=\{a,b,c,d\}$: (i) a Type I qnet $\TA{a}{b}{c}{d}$; 
(ii) a Type II qnet $\TB{a}{b}{c}{d}$; (ii) a Type III qnet $\TC{a}{b}{c}{d}$; 
(iv) a Type IV qnet $\TD{a}{b}{c}{d}$.
Type IV is of cycle type.
}
\label{fig:qnet}
\end{figure}

We now turn to characterizing when a qnet system is displayed by a
level-1 network. To do this, we introduce some additional concepts 
concerning qnet systems.

First, a qnet system $\cF$ on $X$ is {\em consistent} (on subsets of $X$ of size three) if 
for all subsets $A \in {X \choose 3}$, $F|_A$ is isomorphic to $F'|_A$,
for each pair of qnets in $\cF$
with $A\subseteq L(F)\cap L(F')$.  In addition, a 
qnet system $\cF$ on $X$ is {\em minimally dense} if for all $Y \in {X \choose 4}$, there 
exists precisely one qnet in $\cF$ with support $Y$.

Now, we say that a 
qnet system $\cF$ on $X$ is {\em cyclically-transitive or cyclative} if for all subsets 
$\{a,b,c,d,x\} \in {X \choose 5}$ with 
$\{\TD{a}{b}{c}{d},  \TD{x}{a}{c}{d}\} \subseteq \cF$, the system 
$\cF$ also contains $\TD{a}{b}{d}{x}$. Note that this is closely related
to the cyclic-ordering inference rule given in \cite[Proposition 1]{bandelt1992canonical}.
In addition, we say that a qnet system $\cF$ on $X$ is {\em saturated}, if for all subsets
$\{a,b,c,d,x\} \in {X \choose 5}$, the following hold:
\begin{itemize}
\item[(S1)] If  $\cF$ contains $\TA{a}{b}{c}{d}$, then $\TA{a}{b}{c}{x}$, or $\TBB{a}{b}{c}{x}$, or $\TA{a}{x}{c}{d}$, or $\TB{a}{x}{c}{d}$ is contained in $\cF$.\\

\item[(S2)] If $\cF$ contains $\TB{a}{b}{c}{d}$, then $\TB{a}{b}{c}{x}$, or $\TC{a}{b}{c}{x}$, or $\TA{a}{x}{c}{d}$, or $\TB{a}{x}{c}{d}$ is contained in $\cF$.\\

\item[(S3)] If $\cF$ contains $\TC{a}{b}{c}{d}$, then $\TB{a}{b}{c}{x}$, or $\TC{a}{b}{c}{x}$, or $\TBB{a}{x}{c}{d}$, or $\TC{a}{x}{c}{d}$ is contained in $\cF$.\\
\end{itemize}

We next show how these concepts are related. To prove the following result,
given  a qnet system $\cF$, we shall consider the quartet system 
consisting of those quartets that are exhibited by some qnet in $\cF$,
which we shall denote by $\Sigma(\cF)$.

\begin{lemma}
\label{lem:quartet:qnet}
Suppose that $\cF$ is a qnet system on $X$. 

\hspace{2ex}{\rm (i)} If $\cF$ is minimally dense, then $\Sigma(\cF)$ is thin.

\hspace{2ex}{\rm (ii)} If $\cF$ is saturated, then $\Sigma(\cF)$ is saturated.
\end{lemma}

\begin{proof}
For the proof of (i), as $\cF$ is minimally dense, for each subset $Y$ of $X$ with size four, there
exists precisely one qnet $F$ in $\cF$ whose support is $Y$.
Hence there exists at most one quartet in $\Sigma(\cF)$ with support $Y$.

To prove (ii), consider a quartet $Q=ab|cd$ in $\cQ(\cF)$ and 
an arbitrary element $x$ in $X$ that is distinct from $a,b,c,d$.   
Let $F$ be a qnet in $\cF$ such that $Q$ is the quartet exhibited by $F$. 
Then $F$ is Type I, II or III.  
Assume first that $F$ is Type I, then $F=\TA{a}{b}{c}{d}$. Since $\cF$ is saturated, by (S1),  
$$
\{\TA{a}{b}{c}{x}, \TBB{a}{b}{c}{x}, \TA{a}{x}{c}{d}, \TB{a}{x}{c}{d}\} \cap \cQ\not = \emptyset,
$$
and so one of the quartets $ab|cx$ and $ax|cd$
is contained in $\Sigma(\cF)$, as required.
If $F$ is of Type II or III, then
similar arguments using (S2) and (S3), respectively, show that $ab|cx$ or
$ax|cd$ is contained in $\Sigma(\mathcal F)$.
\end{proof}

We now characterize when a minimally dense set of qnets is displayed by a level-1 network.

\begin{theorem}
\label{thm:quarnet}
Let $\cF$ be a minimally dense qnet system on $X$ with $|X|\geq 4$.
Then $\cF=\cF(N)$ for some 
(necessarily unique) binary, level-1 network $N$ on $X$ if and only
if $\cF$ is consistent, cyclative and saturated. 
\end{theorem}

\begin{proof}
Clearly, if $\cF=\cF(N)$ holds for a binary, level-1 network $N$, then $\cF(N)$ is consistent, cyclative and saturated. 

We now show that the converse holds.
Suppose that $\cF$ is a minimally dense qnet system on $X$ that is consistent, cyclative and saturated. 
Consider the quartet system $\Sigma=\Sigma(\cF)$. By Lemma~\ref{lem:quartet:qnet}, 
 $\Sigma$ is thin and saturated. 
Therefore, by Theorem~\ref{thm:quartet:tree},  there exists 
a unique phylogenetic tree $T$ with $\cQ(T)=\Sigma$. 

For each interior vertex $v$ in $T$, let ${\mathcal A}_v$
denote the partition of $X$ induced by deleting $v$ from $T$
so that, in particular, the number of parts in ${\mathcal A}_v$ is
equal to the degree of $v$. 
Note that, for all $A \in {\mathcal A}_v$, if $a \in A$ and $b\in X-A$,  
the path in $T$ between $a$ and $b$ must contain $v$, and if $a,b \in A$, the path between $a$ and $b$ 
does not contain $v$.
 
We next partition the set of interior vertices of $T$. Let $V_1(T)$ be the set of degree-$3$ vertices $v$ in $T$ with the property that there exist three elements, 
one from each distinct part of ${\mathcal A}_v$,  so that there exists a 
qnet $F$ in $\cF$ whose restriction to these three 
elements is of cycle type. Let $V_0(T)$ be the set of degree-$3$ vertices in $T$ not contained in $V_1(T)$. Lastly, let $V_2(T)$ be the set of interior vertices in $T$ with degree at least 4. 


\begin{Claim}\label{one}
A degree-$3$ vertex $v$ in $T$ is contained  in $V_1(T)$ 
if and only if,  for each subset $Y$ of $X$ of size three that contains precisely one element from each part of ${\mathcal A}_v$, the restriction $F|_{Y}$ is of cycle type for every qnet $F$ in $\cF$ with $Y\subset L(F)$. 
\end{Claim}

\begin{proof} 
Since $\cF$ is 
minimally dense, the ``if\," direction follows directly from the definition of $V_1(T)$.

Conversely, let $Y^*=\{a^*_1,a^*_2,a^*_3\}$ be
such that $a^*_i$, $1\leq i \leq 3$, are all contained in distinct parts of ${\mathcal A}_v$
and there exists a qnet $F^*$ in $\cF$ such that $F^*|_{Y^*}$ is of cycle type. 
Now let $Y=\{a_1,a_2,a_3\}$ with $a_i$ all contained in distinct parts of ${\mathcal A}_v$
and let $F$ be an 
arbitrary qnet in $\cF$ with $Y\subset L(F)$. We shall show that $F|_Y$ is of  
cycle type by considering the size of the intersection $Y\cap Y^*$.

First assume that $|Y\cap Y^*|=3$, that is, $Y=Y^*$. Then, as $\cF$ is consistent, $F|_Y$ 
is of cycle type since it is isomorphic to $F^*|_{Y^*}$. 

Second assume that $|Y\cap Y^*|=2$. By swapping the indices, we may further assume 
that $a_1=a^*_1$, $a_2=a^*_2$, and $a_3\not =a^*_3$. In other words, 
we have $Y=\{a^*_1,a^*_2,a_3\}$. Consider $Y'=\{a^*_1,a^*_2,a_3,a^*_3\}$ 
and let $F'$ be the qnet in $\cF$ with $L(F')=Y'$. Since $a_3,a^*_3$ are both 
contained in $A_v$, the quartet $Q'=a^*_1a^*_2|a_3a^*_3$ is contained in $\cQ(T)$. 
As $F'|_{Y^*}$ is of cycle type, this implies that $F'$ is 
either $\TB{a^*_1}{a^*_2}{a_3}{a^*_3}$ or $\TD{a^*_1}{a^*_2}{a_3}{a^*_3}$.
In both cases $F'|_{Y}$ is of cycle type, and hence $F|_Y$ is also of cycle type 
in view of the consistency of $\cF$.

Next assume that $|Y\cap Y^*|=0$. By swapping the indices, we may further
assume that, for $1\leq i \leq 3$, elements $a_i$ and $a^*_i$ are
contained in the same part of ${\mathcal A}_v$ but $a_i\not = a^*_i$.  
Consider 
the sets $Y_1=\{a^*_1,a^*_2,a_3\}$ and $Y_2=\{a^*_1,a_2,a_3\}$, 
and put $Y_0=Y^*$ and $Y_3=Y$. Then we 
have $|Y_i\cap Y_{i+1}|=2$ for $0\leq i \leq 2$. 
Repeatedly applying the argument used when the size of the intersection is two,
it follows that $F|_Y$ is of cycle type, as required. 
 
Lastly, the case $|Y\cap Y^*|=1$ can be established using a
similar argument to that when the size of the intersection is zero. This
completes the proof of the claim.
\end{proof}
 

Although we will not use this fact later, note that it follows from
Claim~\ref{one} that a vertex $v$ in $T$ is contained in $V_0(T)$ if and only if,
for each subset $Y$ of $X$ of size three whose elements are contained in
distinct elements of $\mathcal A_v$, the restriction $F|_Y$ is a tree type
for every qnet $F$ in $\mathcal F$ with $Y\subset L(F)$.


\begin{Claim}\label{two}
Suppose $v \in V_2(T)$. Let $x,y,p,q \in X$ be contained in 
distinct parts $A_x,A_y,A_p,A_q$ of $\mathcal A_v$, respectively. Then
the qnet $F$ in $\cF$ with support $A = \{x,y,p,q\}$  is of Type IV.
Moreover, if $F$ is $\TD{x}{y}{p}{q}$, then,  
for all $x'\in A_x$, $y' \in A_y$, $p'\in A_p$ and $q'\in A_q$, 
the qnet $F'$ with support  $A' = \{x',y',p',q'\}$ is $\TD{x'}{y'}{p'}{q'}$.
\end{Claim}

\begin{proof}
Suppose $F$ is not of Type IV.
Then $\Sigma(F)$ contains precisely one quartet, denoted by $Q$, and $L(Q)=A$.
This implies that $Q \in \Sigma(\cF)=\cQ(T)$. However, $Q$ is not 
contained in $\cQ(T)$ because the path between any pair of 
distinct elements in $A$ contains $v$; a contradiction. Thus $F$ is of Type IV.

Now, suppose $|A\cap A'|=3$. Then 
we may further assume without loss of generality 
that  $x=x'$, $y=y'$, $p=p'$, and $q\not= q'$. Hence $A'=\{x,y,p,q'\}$. Note that the argument in the last paragraph implies that $F'$ is of 
Type IV. If $F'$ is not isomorphic to  $\TD{x}{y}{p}{q'}$, 
then $F'$ is isomorphic to either $\TD{x}{y}{q'}{p}$ or $\TD{x}{p}{y}{q'}$. 
In the first subcase, since $\cF$ is cyclative
and $\{\TD{x}{y}{p}{q}, \TD{x}{y}{q'}{p}\} \subset \cF$, the
qnet $\TD{p}{q}{y}{q'}$ is contained in $\cF$. 
This implies that the quartet $Q'= py|q q'$ 
is not contained in $\cQ(T)$, a contradiction 
since $q,q'$ are contained in $A_q$ while $p,y$ are contained in $X-A_q$. 
The second subcase follows in a similar way.
 
Lastly, if $|A\cap A'|\leq 2$, then note 
that there exists a list of 4-element subsets 
$A=A_0,\cdots,A_t=A'$ for some $t\geq 1$ such that, for $0\leq i <t$, we have 
$|A_i \cap A_{i+1}|=3$ and the two elements in $(A_i-A_{i+1})\cup (A_{i+1}-A_i)$ are contained in the same part of $\mathcal{A}_v$.
Claim~\ref{two} follows by repeatedly applying the argument in the last paragraph to the list. 
\end{proof}

Using the last claim we next establish the following

\begin{Claim} \label{three}
For each vertex $v \in V_2(T)$, there exists a unique circular ordering of 
the parts $A^1,\dots,A^m$ of $\mathcal A_v$ such that, 
for each tuple $A=(a_i,a_j,a_k,a_l) \in A^i \times A^j \times A^k \times A^l$ 
with $1\leq i<j<k<l \leq m$, the qnet 
in $\cF$ with support $\{{a_i,a_j,a_k,a_l}\}$ is 
isomorphic to $\TD{a_i}{a_j}{a_k}{a_l}$. 
\end{Claim}

\begin{proof}  In light of Claim~\ref{two} we can define a quaternary relation
$||$ on the parts of $\mathcal A_v$ by setting $AB || CD$, for all distinct parts $A,B,C,D \in \mathcal A_v$, 
if and only if, for all $x \in A$, $y \in B$, $p \in C$ and $q \in D$, 
the qnet with support  $\{x,y,p,q\}$ is $\TD{x}{y}{p}{q}$.

Now, for all distinct $A,B,C,D,E \in \mathcal A_v$, we show that\\
\noindent
(BD-1): $AB||CD$ implies $BA||CD$ and $CD||AB$;\\
(BD-2): either $AB||CD$, or  $AC||BD$, or  $AD||BC$ (exclusively);\\
(BD-3): $AC||BD$ and $AD||CE$ implies $AC||BE$.

Indeed, let $x \in A$, $y \in B$, $p \in C$, $q \in D$, $r \in E$.
Then (BD-1) holds since $\TD{x}{p}{y}{q}$ is 
isomorphic to $\TD{y}{p}{x}{q}$ and  to $\TD{p}{x}{q}{y}$. Next, 
(BD-2) follows immediately since $\mathcal F$ is minimally dense. 
To see (BD-3) holds, note that since  $AD||CE$ and $AC||BD$ imply
that  $\TD{x}{r}{q}{p}$ and $\TD{x}{q}{p}{y}$ are contained in $\cF$,
using the fact that $\cF$ is cyclative implies that $\TD{x}{r}{p}{y}$
is in $\cF$,  and hence $AC||EB$ holds. Using (BD-1) 
it follows that $AC||BE$, as required.  

Since the quaternary relation $||$ on $\mathcal A_v$ satisfies the  
conditions (BD-1)--(BD-3) as specified in~\cite[Proposition 1]{bandelt1992canonical}, 
it follows that $||$ determines a unique circular ordering of 
the parts in $\mathcal A_v$ as specified in Claim~\ref{three}. 
\end{proof}

Now let $V'=V_1(T)\cup V_2(T)$, and for each vertex $u\in V'$, fix a 
circular ordering of its neighbourhood $N_u(T)$ induced by the ordering of $\mathcal A_u$  
in Claim~\ref{three} if $u\in V_2(T)$, or the necessarily unique circular ordering (clockwise 
and anticlockwise are treated as the same) of $N_u(T)$ if $u\in V_1(T)$ (and hence $|N_u(T)| =3$).
Let $N$ be the level-1 network obtained from $T$ by 
blowing up each vertex $u$ in $V'$ using the given circular ordering of $N_u(T)$. 
We next show that $\cF\subseteq \cF(N)$. 
To this end, fix four  arbitrary elements $a,b,c,d$ in $X$ and let $F$ be the qnet in $\cF$ 
with support $\{a,b,c,d\}$. We need to show that $F\in \cF(N)$. 
There are four cases depending upon whether $F$ is Type I, II, III, or IV. 

First suppose $F$ is of Type I. Without loss of generality, we may assume 
that $F=\TA{a}{b}{c}{d}$. Let $u=\med_T(a,b,c)$. If $u\in V_1(T)\cup V_2(T)$, then $a,b,c$ 
are contained in three distinct parts in the partition ${\mathcal A}_u$ of $X$ 
on $u$. By Claim~\ref{one} and Claim~\ref{two},  it 
follows that $F|_A$ with $A=\{a,b,c\}$ is of cycle type, a contradiction.
Thus $u\in V_0(T)$ and so there exists a 
cut-vertex in $N$ whose removal induces three connected components, 
containing $a$, $b$ and $c$ respectively. 
Similarly, the median $v=\med_T(a,c,d)$ is contained in $V_0(T)$. 
Hence there exists a cut-vertex in $N$ whose removal induces 
three connected components, containing $a$, $c$ and $d$ respectively. 
Let $F'$ be the qnet in $\cF(N)$ whose support is $\{a,b,c,d\}$. 
Thus, by inspecting all possible qnets on $\{a,b,c,d\}$, it follows  
that $F'$ is isomorphic to $\TA{a}{b}{c}{d}$, and hence $F\in \cF(N)$.

Second, suppose that $F$ is of Type II. Without loss of generality, we may 
assume that $F=\TB{a}{b}{c}{d}$. Let $F'$ be the qnet in $\cF(N)$ 
whose support is $\{a,b,c,d\}$. Let $u$ be the median of $a,c,d$ in $T$. 
Then, by an argument similar to the one used in the last paragraph,
it follows that there exists a 
cut-vertex in $N$ (and hence also a cut-vertex in $F'$) whose removal 
results in three connected components, containing $a$, $c$ and $d$ respectively. 
On the other hand, let $v$ be the median of $A=\{a,b,c\}$ in $T$. Then  
$a,b,c$ are contained in three distinct parts of ${\mathcal A}_v$. 
Since $F|_A$ is of cycle type, by Claim~\ref{two} it follows  
that $v\in V_1(T)\cup V_2(T)$, which implies that $F'|_A$ 
is also of cycle type.  Thus,  by inspecting all possible qnets on $\{a,b,c,d\}$,  
it follows that $F'$ is isomorphic to $\TB{a}{b}{c}{d}$, and hence $F\in \cF(N)$.

Next, suppose that $F$ is of Type III. Without loss of generality, we may 
assume that $F=\TC{a}{b}{c}{d}$. Let $F'$ be the qnet in $\cF(N)$ 
whose support is $\{a,b,c,d\}$. Let $u$ be the median of $A=\{a,b,c\}$ 
in $T$ and $v$ be the median of $B=\{a,c,d\}$ in $T$. Since the 
quartet $ab|cd$ is contained in $\cQ(T)$, we know that $u$ and $v$ 
are distinct. Hence, there exists a cut-edge whose deletion puts $a$
and $b$ in one component and $c$ and $d$ in the other connected component.
By an argument similar to that used for analysing when $F$ is of Type II, it follows that  
$F'|_A$ and $F'|_B$ are both of cycle type. 
Hence, by inspecting all possible qnets on $\{a,b,c,d\}$, the qnet $F'$ is isomorphic to 
$\TC{a}{b}{c}{d}$, and hence $F\in \cF(N)$. 

Lastly, suppose that $F$ is of Type IV.  Without loss of generality, we may assume 
that $F=\TD{a}{b}{c}{d}$. Let $F'$ be the qnet in $\cF(N)$ whose 
support is $A=\{a,b,c,d\}$. Hence, there exists no quartet in $\cQ(\cF)$ 
whose support is $A$. Therefore,  
$\med_T(a,b,c)=\med_T(a,b,d)=\med_T(a,c,d)=\med_T(b,c,d)$. 
Denoting this median by $u$, it follows that $u$ is necessarily contained in 
$V_2(T)$, and hence $N_T(u)$ contains $m \ge 4$ vertices.   Now let
$(v_1,v_2,\ldots,v_m)$  be the unique circular 
ordering of vertices $N_T(u)$ induced by the circular 
ordering $A^1,\dots,A^m$ of $\mathcal A_u$ in Claim~\ref{three}.
Without loss of generality, we may assume that $a\in A^1$. Then there 
exists $1<j<k<l\leq m$ such that $(b,c,d)\in A^j \times A^k \times A^l$. 
By the construction of $N$ (which locally is the blow-up at $u$ with respect 
to the circular ordering), it follows that $F'$ is isomorphic to $F$, and hence $F\in \cF(N)$.  

This  shows that $\cF\subseteq \cF(N)$. Since $\cF$ and $\cF(N)$ are 
both minimally dense, we have $\cF=\cF(N)$. Finally, the uniqueness 
statement concerning $N$ is a direct consequence of the uniqueness of $T$ and the unique 
way in which $N$ is constructed from $T$.
\end{proof}

\section{A characterization of level-1 quartet systems} \label{sec:characterize}

We now use Theorem~\ref{thm:quarnet} to characterize when a quartet
system is equal to the set of quartets displayed by a binary level-1
network. This characterization is given as Theorem~\ref{thm:quartet}. 
Let $\cQ$ be a quartet system on $X$.
A quartet $Q$ in $\cQ$ is {\em distinguished} 
 if $Q$ is the only quartet in $\cQ$ 
with support equal to the leaf-set of $Q$.  
Moreover, a network $N$ is called {\rm 3}-cycle free if it does not
contain any cycle consisting of three vertices. 

\begin{theorem}
\label{thm:quartet}
Let $\cQ$ be a dense quartet system  on $X$ with $|X|\geq 4$. 
Then $\cQ=\cQ(N)$ for some binary level-{\rm 1} network $N$ on $X$
if and only if 
the following three conditions hold: 
\begin{itemize}
\item[(D1)] For all $Y\in {X \choose 4}$, we have $m_\cQ(Y)=1$ or  $m_\cQ(Y)=2$. \\
\item[(D2)]  If $\{ab|cd, ad|bc, ax|cd, ac|xd\}\subseteq \cQ$, 
then $\{ab|dx,bd|ax\}\subseteq \cQ$, for $a,b,c,d\in X$ distinct. \\
\item[(D3)] If $ab|cd$ is a distinguished quartet in $\cQ$, then,
  for each $x\in X-\{a,b,c,d\}$ where $a,b,c,d\in X$ are distinct, either $ax|cd$ 
or $ab|cx$ is a distinguished quartet in $\cQ$.  
\end{itemize}
Moreover, if $\cQ$ satisfies {\rm (D1)--(D3)}, then there exists a unique 
level-{\rm 1}, {\rm 3}-cycle free network $N$ with $\cQ(N)=\cQ$.
\end{theorem}

\begin{proof}
  It is easily checked that,  if $\cQ=\cQ(N)$ holds for some binary
  level-1 network $N$, then (D1)--(D3) 
  holds.  Conversely, let $\cQ$ be a dense quartet system
  satisfying (D1)--(D3). Let $\cQ_1 \subseteq \cQ$ 
  be  the set consisting of the distinguished quartets contained in $\cQ$.
  We first associate a phylogenetic $X$-tree $T$  to $\cQ_1$.
If $\cQ_1=\emptyset$, then we let $T$ denote the phylogenetic $X$-tree 
which contains precisely one vertex that is not a leaf (i.e. a``star tree'').
If $\cQ_1 \neq \emptyset$, then let $Q=ab|cd$ be some quartet
contained in $\cQ_1$, $a,b,c,d\in X$. Suppose that there exists some $x \in  X - \{a,b,c,d\}$.
Then by (D3), either $ax|cd \in \cQ_1$ or $ab|cx \in \cQ_1$. It follows
that $\bigcup_{Q \in \cQ_1} L(Q) = X$. Moreover, as $\cQ_1$ is clearly
thin and by (D3) $\cQ_1$ is saturated, it follows by Theorem~\ref{thm:quartet:tree},  
that there exists a phylogenetic $X$-tree $T$ with $\cQ(T)=\cQ_1$.

Now we construct a qnet system $\cF$ as follows. 
Let $\Pi_1$ be the subset of ${X \choose 4}$ consisting of
those $Y$ with $m_\cQ(Y)=1$, and $\Pi_2= {X \choose 4} \setminus \Pi_1$. 
To each  $\pi=\{a,b,c,d\} \in \Pi_1$ we associate a qnet $F(\pi)$ 
as follows. Swapping the labels of the elements in $\pi$ if necessary, we may assume 
that $Q=ab|cd$ is the (necessarily unique) quartet in $\cQ_1$ with leaf-set $\pi$. Now 
let $v_1$ and $v_1'$ be the median of $\{a,b,c\}$ in $Q$ and $T$, respectively. 
Similarly, let $v_2$ and $v_2'$ be the median of $\{a,c,d\}$ in $Q$ and $T$, respectively. 
Then $F(\pi)$ is the qnet on $\{a,b,c,d\}$ obtained from $Q$ by 
performing a blow-up on each of $v_i$, where $i\in \{1,2\}$, if and only if 
the degree of $v'_i$ in $T$ is at least four. 

We also associate a qnet $F(\pi)$ to each  $\pi=\{a,b,c,d\} \in \Pi_2$ 
as follows.  Swapping the labels of the elements in $\pi$ if 
necessary, we may assume that the quartets in $\cQ$ 
with leaf-set $\{a,b,c,d\}$ are $ab|cd$ and $ad|bc$.
We then define $F(\pi)$ to be the qnet $\TD{a}{b}{c}{d}$.  

Now, let $\mathcal F=\{F(\pi): \pi\in \binom{X}{4}\}$. By construction $\cF$ is minimally dense. 
Moreover, $\cQ(\cF)=\cQ$, and $\cF$ is cyclative in view of (D2).  

Next, we shall show that $\cF$ is consistent. Fix a subset $\{a,b,c\} \in {X \choose 3}$ 
and consider its median $v$ in $T$. By construction, it suffices to establish the
claim that the degree of $v$ is three in $T$ if and only if, for each $d\in X-\{a,b,c\}$,
the set $\pi=\{a,b,c,d\}$ is not contained in $\Pi_2$. 

To see that this claim holds
first note that if $v$  has degree three,  then each of the three 
components of $T-\{v\}$ 
contains precisely one element 
in $\{a,b,c\}$.  Without loss of generality, we may assume that element $d$ is contained in 
the connected component containing element $c$. But this implies that $ab|cd$ is a 
quartet in $\cQ(T)$, and hence $\{a,b,c,d\} \in \Pi_1$. 
On the other hand, if $v$ has degree at least four, then there exists an 
element $x\in X-\{a,b,c\}$ such that $x,a,b,c$ belong to four different 
connected components of $T-\{v\}$. Therefore, $\cQ(T)$ 
and $\{ab|cx,ac|bx,ax|bc\}$ are disjoint. This implies that $\pi=\{a,b,c,x\}$ is not contained in $\Pi_1$, and 
so it is contained in $\Pi_2$. This establishes the claim. 

Next, we show that $\cF$ is saturated. We shall show that
(S2) holds; the fact that $\cF$ satisfies (S1) and (S3) can be established by a similar argument.  
Let $\{a,b,c,d\} \in {X \choose 4}$ be a set that  
satisfies the condition in (S2), that is, $\TB{a}{b}{c}{d}$ 
is contained in $\cF$. 
Then $ab|cd$ 
is a quartet in $\cQ_1=\cQ(T)$. Furthermore, put $u=\med_T(a,b,c)$ and $v=\med_T(a,c,d)$, then the degree of $u$ is at least four and the degree of $v$ is three. 
 Now, fix an element $x\in X-\{a,b,c,d\}$. 
If $x$ and $a$ are in the same connected component resulting from 
deleting $v$ from $T$, then $ax|cd$ is a quartet in $\cQ_1$. Since 
the median of $a,c,d$ in $T$ has degree three, by construction 
either $\TA{a}{x}{c}{d}$ or $\TB{a}{x}{c}{d}$ (but not both) is contained in $\cF$.
Otherwise,  $ab|cx$ is a quartet in $\cQ_1$. Since  
the median $u$ of $a,b,c$ in $T$ has degree greater than three, by construction we can conclude that either $\TB{a}{b}{c}{x}$ or $\TC{a}{b}{c}{x}$ is contained in $\cF$ (but not both). 
This completes the verification of (S2). 

It follows that $\cF$ is minimally dense, cyclative, consistent and saturated. By 
Theorem~\ref{thm:quarnet}, there exists a unique binary level-1 network $N$ on $X$such 
that $\cF(N)=\cF$. By construction, it also follows 
that $\cQ(N)=\cQ(\cF(N))=\cQ(\cF)=\cQ$.
The uniqueness statement in the theorem follows from the uniqueness of $N$ and the fact 
that $\cQ(N)=\cQ(N')$ for two binary level-1 networks $N$ and $N'$ if and only 
if $N$ and $N'$ on $X$ differ only by 3-cycles (see e.g. \cite[Lemma 2]{kp14}).
\end{proof}

\section{Quarnet inference rules and closure} \label{sec:closure}

For a quartet system $\cQ$ on $X$, we write $\cQ \vdash ab|cd$ precisely 
if every phylogenetic $X$-tree that displays $\cQ$ also displays $ab|cd$. 
The statement $\cQ \vdash ab|cd$ is known as a {\em quartet inference rule} \cite{SS03}. 
A well-known example of such a rule is
$$
\{ab|cd,ac|de\} \vdash ab|ce 
$$
which leads to the concept of the {\em semi-dyadic closure} $\closure(\cQ)$ of
the set $\cQ$, that is, the minimal set of quartets that contains $\cQ$ 
and has the property that if 
$\{ab|cd,ac|de\} \subseteq \closure(\cQ)$, then $ab|ce\in \closure(\cQ)$. 

In this section, we define analogous concepts for qnets and show that they
have similar properties to those enjoyed by phylogenetic trees.
If  $\cF$ is a qnet system, we write $\cF \vdash F$
for some qnet $F$ if every binary level-1 network that displays $\cF$ 
also displays $F$. Now, let $\oalpha, \obeta,\ogamma$ 
denote symbols in $\{\ominus, \oplus\}$. For example, $a\oalpha b | c \obeta d$ is 
equivalent to $a\ominus b|c \oplus d$ when $\oalpha=\ominus$ and $\obeta = \oplus$.
We introduce three qnet inference rules on $\cF$: 

\begin{itemize}
\item[(CL1):] 
$
\{a\oalpha b | c \obeta d, b\obeta c | d \ogamma e\} \vdash a \oalpha b | c \obeta e
$
for all $\oalpha, \obeta,\ogamma \in \{\ominus, \oplus\}$; \\

\item[(CL2):] 
$\{a \oplus b | c \oalpha d, \TD{a}{c}{e}{b}\} \vdash a\oplus e | c \oalpha d$ and
$\{a \oplus b | c \oalpha d, \TD{a}{c}{b}{e}\} \vdash a\oplus e | c \oalpha d$ and 
$\{a \oplus b | c \oalpha d, \TD{a}{e}{c}{b}\} \vdash a\oplus e | c \oalpha d\,\,$ 
for all $\oalpha \in \{\ominus, \oplus\}$; \\

\item[(CL3):] 
$
\{\TD{a}{b}{c}{d},\TD{e}{a}{c}{d}\} \vdash \TD{a}{b}{d}{e}.
$
\end{itemize}

\noindent
We illustrate two of these rules in Figure~\ref{infer}.

\begin{figure}[ht]
	\centering
	\includegraphics[scale=0.8]{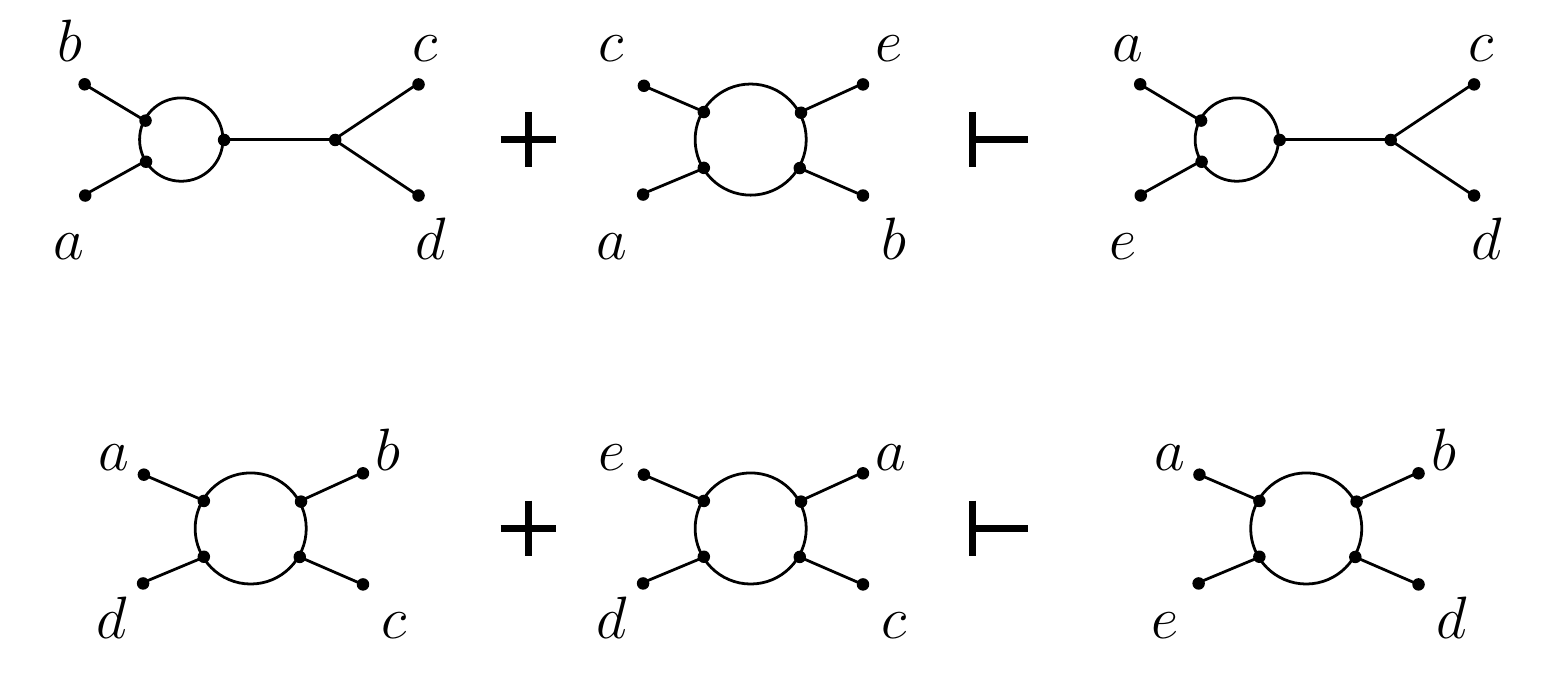}
	\caption{An illustration of the (CL2)  and  (CL3) inference rules. Top:
          The first part of the (CL2) inference rule with $\oalpha=\ominus$.
          Bottom: the (CL3) inference rule.}
	\label{infer}
\end{figure}

Using Theorem~\ref{thm:quarnet}, it is straightforward to show that the above
three rules are well defined. That is, given three qnets $F_1$, $F_2$ and $F$
such that $\{F_1,F_2\} \vdash F$ holds for one of the above three rules, then
every binary level-1 network that displays $\{F_1,F_2\}$ must display $F$. 

For a qnet system $\cF$, we define the set $\closure(\cF)$ to be the minimal
qnet system (under set-inclusion) that contains $\cF$ 
such that if $\closure(\cF) \vdash F$ holds under (CL1)-(CL3), then
$F\in \closure(\cF)$ holds.
We call $\closure(\cF)$ the {\em closure} of $\cF$.

The following key proposition is analogous to that for semi-dyadic closure for quartet systems
(cf. \cite{M83} and \cite[Proposition 2.1]{HMSS}). It follows from the fact that
the closure of a qnet system $\cF$ can clearly be obtained from $\cF$ by
repeatedly applying the qnet rules (CL1)--(CL3) until the sequence of sets so
obtained stabilizes. Note that this process must clearly terminate in polynomial time. 

\begin{proposition}
Let $\cF$ be a qnet system and let $N$ be a binary, level-{\rm 1} network. 
Then $N$ displays $\cF$ if and only if $N$ displays  $\closure(\cF)$.
\end{proposition}

We now show that $\closure(\cF)$ behaves in a similar way to the
semi-dyadic closure of a quartet system (cf.\cite[Exercise 19, p. 143]{SS03}).

\begin{theorem}
\label{thm:closure}
Suppose that $\cF$ is a minimally dense, consistent set of qnets on $X$ with $|X|\geq 5$. 
Then the following statements are equivalent:
\begin{itemize}
\item[(i)] $\cF=\cF(N)$ holds for a (necessarily unique) binary,
  level-{\rm 1} network $N$ on $X$; \\

\item[(ii)] $\closure(\cF)=\cF$;\\

\item[(iii)] For every {\rm 3}-element subset $\cF'$ of $\cF$,
  the subset $\cF'$ is displayed by some binary level-{\rm 1} network on $X$.
\end{itemize}
\end{theorem}

\begin{proof}
  The fact that (i) implies (ii) and  (i) implies (iii) are straightforward. We
  complete the proof by showing that (ii) implies (i) and (iii) implies (i). 

For the proof of (ii) implies (i), suppose that $\closure(\cF)=\cF$. Note first that by (CL3) 
$\cF$ is cyclative. 
Moreover, $\cF$ is minimally dense and consistent by assumption. 
Hence, by Theorem~\ref{thm:quarnet}, it suffices to show that 
$\cF$ is saturated. To this end, let $w,x,y,z,t$ be five pairwise distinct
elements in $X$ such that 
$F=w\oalpha x|y \obeta z$ is contained in $\cF$ with $\oalpha,\obeta \in \{\oplus,\ominus\}$
and $(\oalpha,\obeta)\not =(\ominus,\oplus)$. 
We need to show that $\cF$ satisfies (S1)--(S3). 


For $p \in \{w,x,y,z\}$, let $F_p$ be the qnet  
on $\{w,x,y,z,t\}-\{p\}$ that is contained in $\cF$ (which 
must exist as $\cF$ is minimally dense). First assume 
that there exists some element $p$ in $\{w,x,y,z\}$ such that 
the qnet $F_p$ is of Type IV. Without 
loss of generality, assume $p=w$ (the other cases can be established in a similar manner). 
Since $F_w$ is of Type IV, by the consistency of $\cF$ we have
 $F=y\oplus z|w\oalpha x$. 
Now, applying (CL2) with $a=y$, $b=z$, $c=w$, $d=x$, $e=t$ 
implies that $y\oplus t|w \oalpha x \in\closure(\cF)=\cF$, by (ii).
Therefore,  $\cF$ satisfies (S2) and (S3)
(corresponding, respectively, to taking $\oalpha=\ominus$ and
$\oalpha=\oplus$). 
It follows that in the remainder of the proof we can assume that
none of the qnets in $\{F_w,F_x,F_y,F_z\}$ is 
of Type IV. 


For convenience, in the following, we will use the convention that
when we apply (CL1), we will write a 5-tuple and assume that the 
$i$-th element in the 5-tuple will correspond to 
the $i$-th element in the tuple $(a,b,c,d,e)$ of elements
used in (CL1) for $1\leq i \leq 5$.


To show that $\cF$ satisfies (S1), suppose that $F=\TA{w}{x}{y}{z}$. 
Note first that if $F_x=w\ominus y|z \oalpha t$, then applying (CL1) to 
$(x,w,y,z,t)$ implies $x\ominus w |y \ominus t \in  \closure(\cF)=\cF$, and hence (S1) holds. 
Similarly, if $F_z=w\ominus y|x \oalpha t$, then applying (CL1) to $(z,y,w,x,t)$ 
implies $z\ominus y |w \ominus t \in \cF$, and hence (S1) holds. 
Therefore, if (S1) does not hold, then, by consistency, we may assume  $F_x=w\ominus z|y \oalpha t$ 
and $F_z=x\ominus y|w \oalpha t$ with $\oalpha \in \{\ominus, \oplus\}$. 
Considering $F_x$ and $F_z$, and applying (CL1) to $(x,y,t,w,z)$ implies
$x\ominus y | t \oalpha z \in\cF$. 
On the other hand, considering $F$ and $F_z$ and 
applying (CL1) to $(z,y,x,w,t)$ implies that $z\ominus y | x \ominus t \in \cF$, 
a contradiction to the fact that $\cF$ is minimally dense. Thus $\cF$ satisfies (S1). 
 
Using an argument similar to the one that we used to show that
$\cF$ satisfies (S1), it is straightforward to deduce that $\cF$ satisfies (S2) and (S3).

We next prove that (iii) implies (i).  Since $\cF$ is minimally dense and consistent by assumption, 
it follows by Theorem~\ref{thm:quarnet} that it suffices to show that $\cF$ is cyclative and saturated.

First we show that $\cF$ is cyclative. If not, then there exists 
five elements $Y=\{w,x,y,z,t\}$ such that $F_1=\TD{w}{x}{y}{z}$ and $F_2=\TD{t}{w}{y}{z}$ 
are contained in $\cF$ but $F=\TD{w}{x}{z}{t}$ is not contained in $\cF$. 
Let $F'$ be the (necessarily unique) qnet in $\cF$ whose leaf 
set is $\{w,x,z,t\}$. Then $F'\not = F$. Consider the set $\cF'=\{F',F_1,F_2\}$.
The assumption (iii) implies that $\cF'$ is displayed by a binary
level-1 network $N$ on $X$. Consider $N'=N|_Y$. Then $\cF'\subseteq \cF(N')$. 
By Theorem~\ref{thm:quarnet}, $\cF(N')$ is minimally dense and cyclative. 
Since $\{F_1,F_2\}\subseteq \cF(N')$, it follows that $F\in \cF(N')$, a contradiction 
in view of $F' \in \cF(N')$.

Second we show that $\cF$ is saturated. Here we only 
show that $\cF$ satisfies (S2) as showing that $\cF$
satisfies (S1) and (S3) can be done in a similar manner. If $\cF$ does not satisfy (S2),  
then there exists a $5$-element set $Y=\{w,x,y,z,t\}$ such 
that $F=\TB{w}{x}{y}{z}$ is contained in $\cF$ while, for the qnet system 
$$
\cF^* = \{\TB{w}{x}{y}{t}, \TC{w}{x}{y}{t}, \TA{w}{t}{y}{z}, \TB{w}{t}{y}{z} \},
$$
we have $\cF^*\cap \cF=\emptyset$.  Let $F_1$  and $F_2$ be the qnets 
in $\cF$ with leaf sets $A=\{w,x,y,t\}$ and $B=\{w,t,y,z\}$, respectively
which must exist as $\cF$ is minimally dense by assumption.  
Then neither $F_1$ nor $F_2$ is contained in $\cF^*$.  

Lastly, consider the subset $\cF'=\{F,F_1,F_2\}$ of $\cF$. Then 
as assumption (iii) holds it follows
that $\cF'$ is displayed by a binary level-1 network $N$ on $X$. 
Consider $N'=N|_Y$. Then $\cF'\subseteq \cF(N')$. 
By Theorem~\ref{thm:quarnet}, $\cF(N')$ is minimally dense and saturated. Using 
the fact that $\cF(N')$ is saturated, it follows that  $\cF^*\cap \cF(N') \not = \emptyset$
as $F \in \cF(N')$. 
Therefore,  $\cF(N')$ contains either two distinct qnets 
on $A$ or two distinct qnets on $B$, a contradiction to the fact that $\cF(N')$ is minimally dense.
Thus (iii) implies (i), thereby completing the proof of the theorem. 
\end{proof}

Note that it follows from Theorem~\ref{thm:closure} that we can decide whether or
not a given minimally dense set of qnets $\cF$ is displayed by a level-1 binary phylogenetic 
network in polynomial time since, as observed above, we can compute 
$\closure(\cF)$ in polynomial time.

\section{Discussion} \label{sec:discuss}

We have shown that by considering quarnets we can 
define natural inference rules, as well as the concept of quarnet closure. With 
quartets there are various types of inference rules, which imply 
alternative definitions of closure for quartet systems (see e.g. \cite{BS95,SS03}).
It would thus be of interest to explore whether there are other 
types of inference rules for quarnets and, if so, what their properties are.
In this paper, we have focused on understanding the closure  for a
minimally dense set of quarnets. 
For real data, there can be cases where it may be necessary to consider
non-minimally dense sets (e.g.
in case there is missing data). Hence it could be useful to develop
results for such 
situations. However, it should be noted that understanding the closure
of a non-minimally dense set 
quartets is already quite challenging (for example, as opposed to the
minimally dense case, 
deciding whether or not an arbitrary set of quartets can be displayed by a 
phylogenetic tree is NP-complete \cite{S92}).
	
In many applications, biologists prefer to use weighted phylogenetic
trees and networks to model their data, 
where non-negative numbers are assigned to edges of the tree or
network to, for example, represent evolutionary 
distance. The problem of considering when a dense set of weighted quartets can be
represented by a weighted phylogenetic tree has been considered
in \cite{D03,G08}. Given the results in this paper, it could
therefore be of interest to consider how weighted level-1 networks may be inferred from
dense sets of weighted quarnets. In applications, it can also be useful to
consider rooted networks, which are essentially leaf-labelled, directed acyclic graphs.
Edges in such networks have a direction which represents  the fact that species evolve through
time from a common ancestor (represent in graph theoretical terms by a root vertex).
For such networks, the concept of level-1 networks can be defined in a similar way to the
unrooted case, and algorithms are known for deciding when minimally
dense collections of 3-leaved, level-1
rooted phylogenetic networks (which are known as {\em trinets})
can be displayed by a single
phylogenetic network \cite{HM13,H17}. It would
thus be of interest to consider inference rules for trinets. Moreover, for both the rooted
and unrooted case, it could be worth exploring whether there are inference rules for
more complicated networks (e.g. networks with level higher than one, as defined in e.g. \cite{G12}). 
Although results in  \cite{im17} indicate that such inference rules might exist, if they do,  then 
we expect that  these will probably be quite complicated.\\

\noindent {\bf Acknowledgement}
KTH, VM, and CS thank the London Mathematical Society for its support.
CS was supported by the New Zealand Marsden Fund.

\end{document}